\documentclass[11pt]{article}

\usepackage[a4paper,margin=0.80in,footskip=0.25in]{geometry}

\usepackage{makeidx}  
\usepackage{cite}
\usepackage{amsmath,amssymb}
\usepackage{amsthm}
\usepackage{kbordermatrix}
\usepackage{framed}
\usepackage{array}
\usepackage{graphicx} 
\usepackage{float}
\usepackage{url}
\usepackage{multirow}  
\usepackage{color}    
\usepackage{subfig}
\usepackage{hyperref}
\usepackage{enumitem}
\usepackage{blkarray}
\usepackage{algorithm}
\usepackage[noend]{algpseudocode}
\usepackage{etoolbox}
\let\bbordermatrix\bordermatrix
\patchcmd{\bbordermatrix}{8.75}{4.75}{}{}
\patchcmd{\bbordermatrix}{\left(}{\left[}{}{}
\patchcmd{\bbordermatrix}{\right)}{\right]}{}{}

\setlistdepth{9}

\setlist[itemize,1]{label=$\bullet$}
\setlist[itemize,2]{label=$-$}
\setlist[itemize,3]{label=$*$}
\setlist[itemize,4]{label=$+$}
\setlist[itemize,5]{label=$.$}
\setlist[itemize,6]{label=$\bullet$}
\setlist[itemize,7]{label=$\bullet$}
\setlist[itemize,8]{label=$\bullet$}
\setlist[itemize,9]{label=$\bullet$}

\renewlist{itemize}{itemize}{9}

\newcounter{linecounter}
\newcommand{\linenumbering}{\ifthenelse{\value{linecounter}<10}
{(0\arabic{linecounter})}{(\arabic{linecounter})}}

\renewcommand{\thelinecounter}{\ifnum \value{linecounter} >
9\else 0\fi \arabic{linecounter}}

\newcommand{\toto}{xxx}

\newif\ifshowcomments
\showcommentstrue

\ifshowcomments
\newcommand{\mynote}[2]{\fbox{\bfseries\sffamily\scriptsize{#1}}
 {\small$\blacktriangleright$\textsf{\emph{#2}}$\blacktriangleleft$}}
\else
\newcommand{\mynote}[2]{}
\fi

 \makeatletter
\newtheorem*{rep@theorem}{\rep@title}
\newcommand{\newreptheorem}[2]{%
\newenvironment{rep#1}[1]{%
 \def\rep@title{#2 \ref{##1}}%
 \begin{rep@theorem}}%
 {\end{rep@theorem}}}
\makeatother

 \newtheorem{defin}{Definition}
 \newtheorem{lemma}{Lemma}
\newtheorem{theorem}{Theorem}

\newtheorem{corollary}{Corollary}
\newreptheorem{lemma}{Lemma}
\newreptheorem{theorem}{Theorem}

\begin{document}

\title{Investigating the Cost of Anonymity on Dynamic Networks \footnote{A shorter version of this manuscript has appeared as brief announcement at PODC 2015.}}

\author{ 
  Giuseppe Antonio {\sc Di Luna} and Roberto {\sc Baldoni} \\
 Dipartimento di Ingegneria Informatica, Automatica e Gestionale Antonio Ruberti\\    Universit\`a degli Studi di Roma La Sapienza\\Via Ariosto, 25, I-00185 Roma, Italy\\
  \texttt{\{baldoni,  diluna\}$@$dis.uniroma1.it}}

\date{}
\maketitle
\thispagestyle{empty}

\begin{abstract}
In this paper we study the difficulty of counting nodes in a synchronous dynamic network where nodes share the same identifier, they communicate by using a broadcast with unlimited bandwidth and, at each synchronous round,    network topology may change. To count in such setting, it has been shown that the presence of a leader is necessary. 
We focus on a particularly interesting subset of dynamic  networks, namely  \textit{Persistent Distance} - ${\cal G}($PD$)_{h}$, in which each node has a fixed distance from the leader across rounds and such distance is at most $h$. In these networks the dynamic diameter $D$ is at most $2h$. 
We prove the number of rounds for counting in ${\cal G}($PD$)_{2}$ is at least logarithmic with respect to the network size $|V|$. Thanks to this result, we show that counting on any dynamic anonymous network with $D$ constant w.r.t. $|V|$  takes at least $D+ \Omega(\text{log}\, |V| )$ rounds where $\Omega(\text{log}\, |V|)$ represents the additional cost to be payed for handling anonymity.  At the best of our knowledge this is the fist non trivial, i.e. different from $\Omega(D)$, lower bounds on counting in anonymous interval connected networks with broadcast and unlimited bandwith.

\end{abstract}

\pagenumbering{arabic}

\newpage
\section{Introduction}


Computing over a dynamic distributed system has become mainstream in the recent years, 
this has been due to advent of peer-to-peer systems, the capillary distribution  of mobile devices and growing impact of sensors networks. Such technologies force the designer of distributed systems to think about the dynamicity as a fundamental and persistent  property of the system that is present during the entire system lifetime. 


As in \cite{surveysantoro,kempe,dell2007,LKM06,1806760}, in this paper we consider  systems where processes are stable (i.e., the set $|V|$)  while there is  an adversary that continuously changes the underlying communication graph connecting such processes.  The adversary can be either fair or worst-case. A fair adversary creates or removes edges from the communication graph  following a strategy that does not aim to violate the correctness of the distributed algorithm (e.g., random strategy). This is the typical dynamic behavior exhibited by a peer to peer system \cite{kempe,LKM06}. 
A \textit{worst-case adversary} continuously changes the underline communication graph, having access to nodes' local variables, in order to deploy at \textit{each computational round} the worst possible network topology to contrast the correctness of the algorithm.
The adversary is constrained to maintain in each round a connected topology, i.e. $1$-interval connectivity property \cite{1806760}.

\textit{We investigate the difficulty of counting the network size under a worst-case adversary, where anonymous nodes communicate any amount of data among each other by using an anonymous broadcast}. Solving the counting problem is of primary importance since it may be seen as a basic building block necessary, for example, to compute generic aggregated function. It has been shown in \cite{MichailCS12} that  counting cannot be solved without a leader and several counting algorithms have been provided in the last couple of years that minimize the knowledge about the adversary that nodes need in order to count (e.g., \cite{ICDCS}). However, at the best of our knowledge there is no study that addresses the actual cost of anonymity in terms of time necessary to output the number of processes that form the computation and no trivial lower bound, i.e. different from $\Omega(D)$, is known for such environment.

To study this cost, we first identify an interesting family of subsets of dynamic graphs,  namely ${\cal G}($PD$)_{h}$ graphs. In such graphs  each node has a persistent distance from the leader across rounds and such distance is at most $h$. 
 It is easy to see that graphs in ${\cal G}($PD$)_{1}$ are actually star graphs with the leader at the center and that the adversary cannot change any of such graphs without compromising the connectivity of the graph itself.  In such a scenario the leader is able to output the exact count in one round independently of the anonymity of the processes.

 Let us now consider ${\cal G}($PD$)_{2}$, processes at distance 2 are connected to the leader by an unknown number of paths that  dynamically changes at each round. If nodes are  anonymous, ambiguity is created among these multiple dynamic paths. To overcome such ambiguity,  the leader will need more than  2 rounds \footnote{Solutions exploiting randomness (i.e. tossing coins hoping for different outcomes) are not viable, since we  assume the source of randomness available to processes is governed by the worst case adversary.}. For such family of graphs we show a lower bound, $\Omega(\log\,|V|)$, on the number of rounds necessary to solve the counting problem. This result is proved by showing that counting in a dynamic labeled multigraph where nodes are directly connected to the leader is at most as difficult as counting in ${\cal G}($PD$)_2$. Then we prove the bound on such dynamic multigraphs by using linear algebra techniques.

The rest of the paper is structured as follows: Section 2 presents the related work; Section 3 defines the system model; Section \ref{lowerboundcount} shows the bound on ${\cal G}($PD$)_2$.
Section 5 concludes the paper.

\section{Related Work}

The question concerning what can be computed  on top of static anonymous networks, has been pioneered by Angluin in \cite{An80} and  by Yamashita and Kameda \cite{Yamashita:1988:CAN:62546.62568}. In the domain of non-anonymous dynamic networks the counting problem has been addressed in the following  contexts: (i) dynamicity governed by node churn in the context of distributed query execution \cite{BGGM07,BaldoniBCQ12}, (ii) dynamicity governed by random adversary in the context of peer-to-peer networks  \cite{LKM06} and (iii) dynamicity governed worst-case adversary in the context of $1$-interval connectivity \cite{dell2007,1806760}.




\smallskip

\noindent
 {\bf Counting in anonymous dynamic networks}: \color{black}
  In \cite{kempe}, the authors propose a gossip-based protocol to compute aggregation function in a dynamic network by exploiting an invariant, called \emph{conservation of mass}, defined over the whole set of processes. The network graph considered by \cite{kempe} is  governed by a fair random adversary. The first work investigating the problem of counting in an anonymous network with worst-case adversary is  \cite{MichailCS12}.  
The authors provided an algorithm that, under the assumption of a fixed upper bound on the maximum node degree, it computes an upper bound on the size of the network.   Building on this result, \cite{DBBC14} proposes an exact counting algorithm under the same assumption.  Finally,  \cite{ICDCS}  provides  a counting algorithm for 1-interval connected networks considering each node is equipped with a local degree detector, i.e. an oracle able to predict the degree of the node in each graph generated by the adversary. Both \cite{DBBC14} and \cite{ICDCS} terminate in an exponential number of rounds. 
  
\smallskip

\noindent
{\bf Bounds on adversarial-based dynamic networks}: 
 A fundamental problem that is correlated with counting in dynamic networks with IDs is the $k$ tokens dissemination, defined as follows \cite{1806760}:  each token is initially owned by a node belonging to $V$, then processes exchange tokens, the $k$ tokens dissemination terminates when all $k$ tokens have been received by each node in $V$. \cite{1806760} proved that when each node may send only one token at each round any $k$ token dissemination algorithm based on token forwarding  \footnote{Token-forwarding algorithms are not allowed to combine, split, or change tokens in any way \cite{disc12}.}  terminates in $\Omega(|V| \text{log}\, k)$ round.  
 In \cite{soda13}, the authors improved the bound to $\Omega (\frac{ |V| k}{\text{log}\,|V|})$.  Starting from  these results,  \cite{disc12} provides   bounds for different adversarial-based dynamic networks. It is well known that in network with IDs  $n$ (all-to-all) token dissemination solves counting \cite{opodis13}.  In the same paper, it is introduced a connection between two-party token dissemination, a variant of $k$-token dissemination, and the counting problem. The authors show a lower bound for the two-party problem is also a lower bound for counting. In anonymous dynamic networks considered in  \cite{opodis13}, $k$-token dissemination can be solved by a trivial flooding algorithm in ${\cal O}(D)$ rounds.
Finally,  \cite{directeddataaggregation} shows that in directed static network with IDs and limited  bandwidth, the number of rounds needed to solve counting is function of the network size even when $D=2$.


\section{Model of the computation}

We consider a synchronous distributed system composed by a finite static set of processes $V$ (also called \emph{nodes}). Nodes in $V$ are \emph{anonymous}, i.e., they initially have no identifiers and execute a deterministic \emph{round-based} computation. Processes communicate through a communication network which is  \emph{dynamic}. We assume at each round $r$ the network is stable and represented by a graph $G_r=(V,E(r))$ where $V$ is the set of nodes and $E(r)$ is the set of bidirectional links at round $r$ connecting processes in $V$. 

\begin{defin}\label{def-DG} 
A \emph{dynamic graph} $G=\{ G_0,G_1,\ldots, G_r,\ldots \}$  is an infinite sequence of graphs one at each round $r$ of the computation. 
\end{defin}

 The neighborhood of a node $v$ at round $r$ is denoted by $N(v,r)=\{v':\{v',v\}\in E(r)\}$. Given a round $r$ we denote with $p_{v,v'}$ a path on $G_r$ between $v$ and $v'$. Moreover we denote with $P(r)_{v',v}$, the set of all paths between $v,v'$ on graph $G_r$. The distance $d_{r}(v',v)$ is the minimum length among the lengths of the paths in $P(r)_{v',v}$, the length of the path is defined as the number of edges.

 Computation proceeds by  rounds. Every round is divided in a \emph{send phase} where processes send  the messages for the current round and a \emph{receive phase} where nodes process received messages and prepare those that will be sent in the next round.
Processes communicate with its neighbors through an \emph{anonymous broadcast} primitive: a message $m$ sent by node $v_i$ in the send phase of round $r$ will be delivered to all its neighbors during the receive phase of $r$. Let us remark that in such model a node does not know the value $|N(v,r)|$ before the receive phase of round $r$.
 
   Let us define the dynamic diameter $D$ of dynamic graph $G$. A node $v$ floods message $m$ by broadcasting it at each round, each process receiving a flooded message $m$ starts, in its turn, a flooding of $m$. The flood of $m$ terminates when it has been received by all nodes. A network has a dynamic diameter $D$ if for any $v$ and for any round $r$ the flood of a message that starts at round $r$ from node $v$ terminates at most by round $r+D$. Intuitively the dynamic diameter is the maximum time needed to disseminate messages to all nodes in the network.
  In this work we consider networks where $D$ is fixed and constant with respect to $|V|$.

\paragraph{Leader-based computation and worst case adversary} 
We assume the selection of a topology  graph at round $r$ is done by an omniscient adversary that may choose at each step the worst configuration to challenge a counting algorithm. Due to the impossibility result shown in \cite{MichailCS12}, we assume any counting algorithm that works over a dynamic graph has a leader node $v_l$  starting with a different unique state w.r.t. all the other nodes. 

 \begin{defin}\label{def-Counting}
 Given a dynamic network $G$ with $|V|$ processes, a distributed algorithm ${\cal A}$ solves the counting on $G$ if it exists a round $r$ at which the leader outputs $|V|$ and terminates. 
 \end{defin}

 \vspace{-0.3cm}

\paragraph{Persistent distance dynamic graphs} 
Let us characterize dynamic graphs according to the distances among a  node $v$ and the leader $v_l$.

\begin{defin}\label{def-safe} (Persistent Distance over $G$)
Let us consider a dynamic graph $G$.
The persistent distance between $v$ and $v_l$ over $G$, denoted $D(v,v_l)=d$, is defined as follow:
 $D(v,v_l)=d$ iff $\forall r, d_{r}(v,v_l)=d$.
\end{defin}


Let us now introduce a set of dynamic graphs based on the distance between the leader and the nodes of a graph.
 
\begin{defin}\label{def-PD} (Persistent Distance set)
A graph $G$ belongs to  Persistent Distance set, denoted ${\cal G}($PD$)$ , iff   $\forall v \in G, \,\exists d \in \mathbb{N}^{+}   :: $D$(v,v_l)=d$
\end{defin}


\paragraph{Graphs in ${\cal G}($PD$)_2$ }
Among the dynamic graphs belonging to ${\cal G}($PD$)$ we can further consider the set of graphs, denoted ${\cal G}$(PD$)_h$,  whose nodes have maximum distance $h$ from the leader with $1<h \leq |V|$. Thus, given a graph in ${\cal G}$(PD$)_h$  we can partition its nodes in $h$ sets, $\{V_0,V_1,\ldots,V_h\}$, according to their distance from the leader. The focus of this paper is on dynamic graphs belonging to ${\cal G}($PD$)_{2}$.  As an example, figure \ref{figure:gpd2} depicts a graph belonging to  ${\cal G}($PD$)_2$ at round $0, 1$ and $2$ whose dynamic diameter is $D=4$. If node $v_0$ starts a flood at round $0$, this flood will indeed reach node $v_3$ at round $3$. The task of the leader node $v_l$ is to count nodes in $V_2$.

\begin{figure}[H] 
\begin{center}
\includegraphics[scale=0.8]{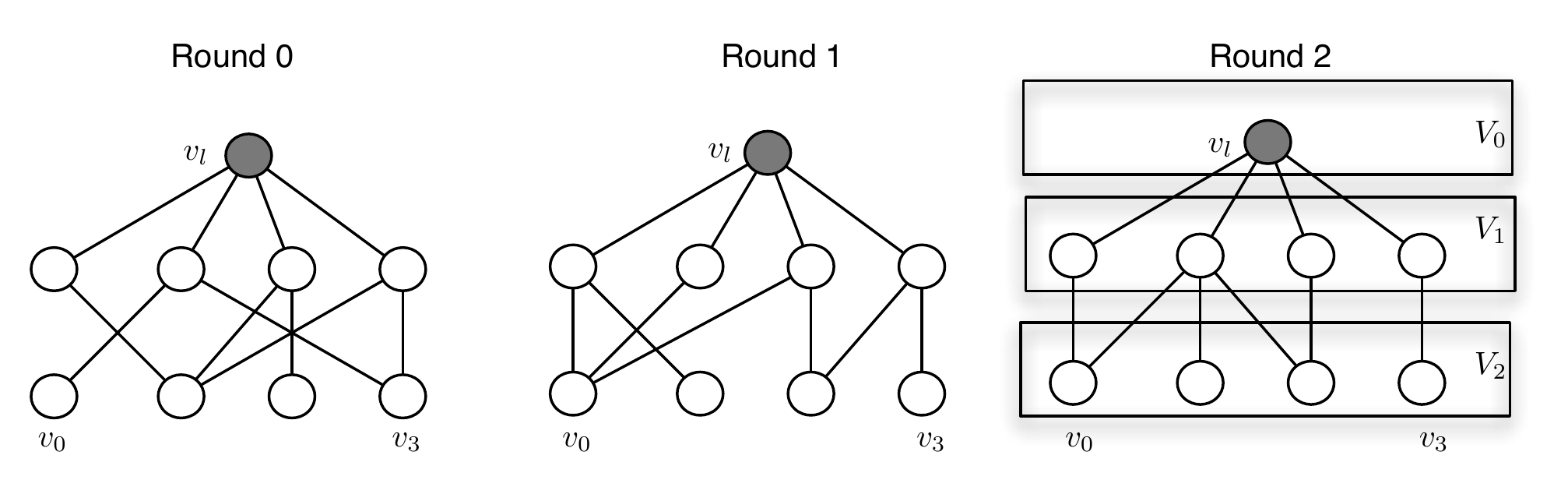}

\caption{An example of a graph belonging to ${\cal G}($PD$)_2$ along three rounds \label{figure:gpd2} }

\end{center}

\end{figure}

\section{Lower bound for ${\cal G}$(PD$)_2$ \label{lowerboundcount}}
In this section we consider the ${\cal G}$(PD$)_2$ set and compute a lower bound for counting time. This is done by introducing a \emph{Dynamic Bipartite Labeled $k$-Multigraphs} (${\cal M}($DBL$_{k})$),  by showing that counting on ${\cal G}($PD$)_2$ requires at least the same number of rounds as counting over  ${\cal M}$(DBL$_k)$ and by finally showing a lower bound on the number of rounds needed to count over ${\cal M}($DBL$_{k})$.


\subsection{Counting in Dynamic Bipartite Labeled $k$-Multigraphs (${\cal M}($DBL$_{k})$)} 
Let consider a dynamic connected multigraph $M$ defined as follows $M=\cup^{\infty}_{r=0}\{(\{v_l\} \cup W,E(r),f_r,l_r)\}$ where $E(r)$ is a set of edges at round $r$, $W$ a set of nodes, $f_r: E(r) \rightarrow \{v_l\} \times W$ a function that maps each edge to the endpoints nodes and $l_r: E(r) \rightarrow \{1,2,\ldots,k\}$ a function labeling edges. 
$M$ belongs to ${\cal M}($DBL$_{k})$ if for each round $r$ the number of edges connecting a node $v \in W$ to $v_l$ is less than $k+1$, more formally  $\forall r, \forall v \in W, E^{v}(r)=f_r^{-1}(v_l,v) ::  1 \leq |E^{v}(r)| \leq k  $;  Given $ e',e'' \in E^{v}(r)$ we have $l_r(e') \neq l_r(e'')$, that is if two edges $e',e''$ share the same non leader node as endpoint they must have different label, as example see edges that involve node $v$ in Figure \ref{figure:trasf}. For simplicity we will refer as $M_{r}$ the instance of $M$ at round $r$. 
Figure \ref{figure:trasf} shows an example of a dynamic connected multigraph $M$ at round $r$ belonging to ${\cal M}($DBL$_{3})$. We assume that when a node $v \in \{v_l\} \cup W$ receives a message from a node $w$ at round $r$ by edge $e$, it also obtains the label $l_r(e)$.

\begin{figure}[htbp] 
\begin{center}
\includegraphics[scale=0.8]{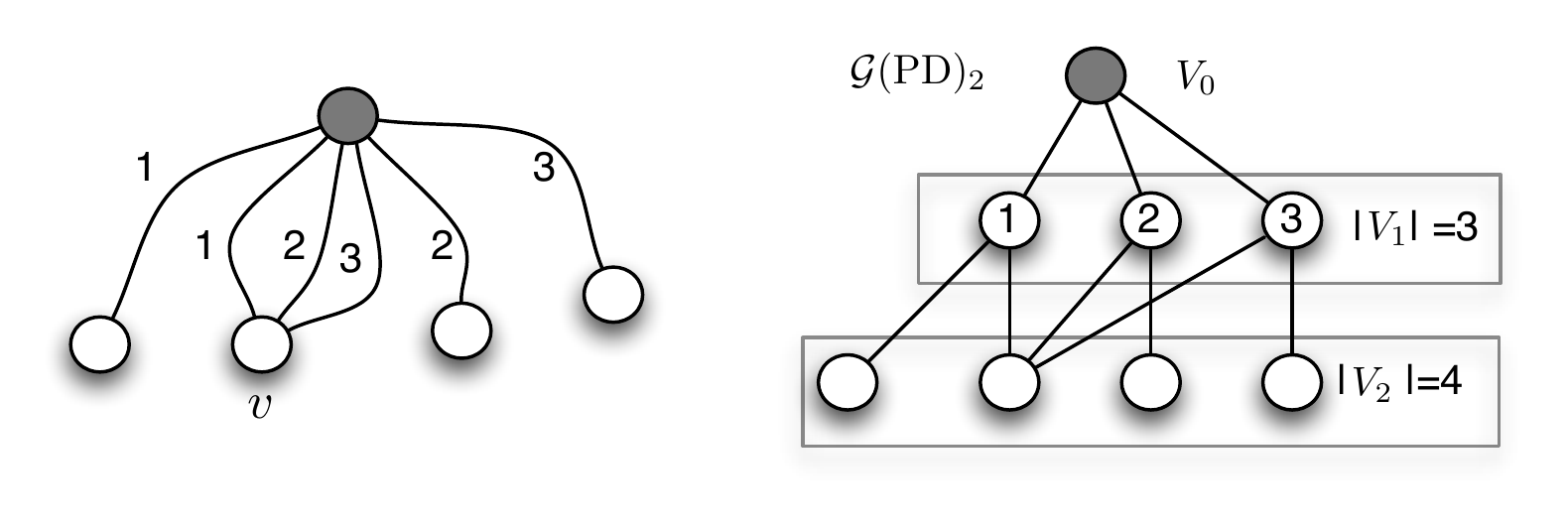}

\caption{Trasformation, at round $r$, from ${\cal M}(DBL_{3})$ multigraph to ${\cal G}(\text{PD})_{2}$. \label{figure:trasf} }

\end{center}

\end{figure}

\begin{lemma} \label{lemma:countingreduction}
Let us consider a dynamic connected multigraph $M$ in ${\cal M}($DBL$_{k})$. If any counting algorithm based on message passing takes more than $T$ rounds to complete on $M$, then there exists a graph $G$ in ${\cal G}$(PD$)_2$ such that any counting algorithm based on message passing requires more than $T$ rounds to complete on $G$. 
\end{lemma}
\begin{proof}
 From $M_r=(\{v_l\} \cup W,E(r),f_r,l_r)$ we build an instance $G^{id}_r=(V= \{(V_0=\{v_l\}) \cup V_1 \cup (V_2=W)\},E_{id}(r))$ belonging to $G^{id} \in {\cal G}$(PD$)_2$ such that $V_1$ contains $k$ nodes  having unique identifiers in $[1,\ldots,k]$, $V_0$ contains only the leader node $v_l$ and the set of nodes $V_2=W$.  Additionally, 
at round $r$ we have that $\exists e:(v,w) \in E_{id}(r)$ with $v \in V_1$ and $w \in V_2$ where $id(v)=j$ if and only if $\exists e' \in E(r)$ with $f_r(e')=(v_l,w), l_r(e')=j$ with $w \in W$. Figure \ref{figure:trasf}
shows the transformation at round $r$ between a dynamic graph in ${\cal M}($DBL$_{3})$ and one in ${\cal G}$(PD$)_2$. Let us notice that the node with label $1$ in $V_1$ at each $G^{id}_{r}$ is connected to the nodes in $V_2$ that correspond to nodes in $W$ that are connected in $M_{r}$ to $v_l$ by edges labeled with $1$.  As a consequence, the leader $v_l$ in $M$ is actually the union of local memories of processes in $\{v_l\} \cup V_1$ in $G^{id}$. Let us assume that there not exist a message passing algorithm solving the counting problem in $M$ with $T$ rounds, then it is not possible to count nodes in $V_2$ on $G^{id}$ in $T$ rounds even by merging the memories of $\{v_l \} \cup V_1$, by knowing the size $k$ of $V_1$ and by having unique IDs for nodes in $V_1$.  Consider now the dynamic graph $G$ derived by $G^{id}$  removing the identifiers of nodes in $V_1$. Counting nodes in $G$ is at least as hard as counting nodes in $G^{id}$. As an example, without identifiers the leader cannot realize if messages of two successive rounds arrive from the same node of $V_1$. Thus it is not possible for the leader to count the size of $G$ in less than $T$ rounds.
\end{proof}

 From the lemma follows that a lower bound for counting on ${\cal M}($DBL$_{k})$ holds also for graphs in ${\cal G}$(PD$)_2$.
Now we introduce some definitions on $M$. Let consider an instance $M$ of the family ${\cal M}($DBL$_{k})$. 

\begin{defin}\label{def-L} (Set of edge labels of a node at round $r$)
Given a node $v \in W$ at round $r$ we define the set of edge labels $L(v,r):\{l_1,\ldots,l_j\}$ with $l_i \in L(v,r)$ iff $\exists e \in E(r)$ and $l_r(e)=l_i$ and $f_r(e)=(v,v_l)$.
\end{defin}

As an example in Figure \ref{figure:trasf}, the edge label set of node $v$ at round $r$ is $\{1,2,3\}$. 

\begin{defin}\label{def-STATE} (State of a non-leader process)
Given a node $v \in W$ at round $r$,  we define the state $S(v,r)$ as an  ordered list $S(v,r):[(\bot),L(v,0),\ldots,L(v,r-1)]$ where $(\bot)$ is the first state of any non-leader node\footnote{For simplicity  whenever not necessary we omit the presence of $(\bot)$ as first element of $S(v,r)$.}. 
\end{defin}

Given a list $A: [L_0,L_1,\ldots,..,L_{r-1}]$ we have that $|A|$ denotes the number of nodes with the same state $S(v,r)=A$ at round $r$. Ref. Figure 1: we have $S(v,r+1)=[\bot,\ldots,\{1,2,3\} ]$ and $|S(v,r+1)|=1$ since $v$ is the only node connected to $v_l$ by $\{1,2,3\}$ at round $r$.

\begin{defin}\label{def-LeaderL} (State of a leader node at round $r$)
Given the leader $v_l$ at round $r$ we define the leader state $S(v_l,r)$ as $[C(v_l,0),\ldots,C(v_l,r-1)]$ where  $C(v_l,i)$ with $i<r$ is  a multiset of elements, $(j,S(v,i)) \in C(v_l,i)$ iff it exists a node $v$ with state $S(v,i)$ connected to $v_l$ by an edge with label $j$. 
\end{defin}

As for states of local nodes,  $|(j,S(v,r))|$ denotes  the number of nodes with state equal to $S(v,r)$ connected to $v_l$ by an edge with label $j$ at round $r$. Let us remark that the state of the leader $v_l$ can be constructed by a simple message passing protocol where  at each round each node sends to the leader its own state and where the leader node sends just  a dummy message.


%

\subsection{Lower Bound for ${\cal M}($DBL$_{k})$}
We introduce some notation on vectors and matrices used in this section.

\paragraph{Linear algebra notation} Given a vector $\mathbf{a} \in \mathbb{Z}^{n}$, we denote as $(\mathbf{a})_j$ the $j$-th component of $\mathbf{a}$ (with $1 \leq j \leq n$) and as $\sum \mathbf{a}$ the sum of all components of $\mathbf{a}$. Additionally, $\sum^{+} \mathbf{a}$ (resp. $\sum^{-} \mathbf{a}$)  denotes  the sum of only the positive (resp.  negative) components of  $\mathbf{a}$. Given two vectors $\mathbf{a},\mathbf{b}$ we have $\begin{bmatrix} \mathbf{a} \\ \mathbf{b} \end{bmatrix}$ is the vector obtained by appending the elements of the vector $\mathbf{b}$ after the last element of vector $\mathbf{a}$.  Given a matrix $\mathbf{M} \in \mathbb{Z}^{n , m}$ we denote with $(\mathbf{M})_j $ its $j$-th row (with $1 \leq j \leq n$) and we denote as $ker(\mathbf{M})$ the set of vectors $\mathbf{a}\in \mathbb{Z}^{m}$ such that $\mathbf{M}\mathbf{a}=\mathbf{0}$. 
We also consider the set of vectors $B=\{\mathbf{a}^{1},\ldots,\mathbf{a}^{\ell}\}$ that form a basis for $ker(\mathbf{M})$, i.e.,  
$ker(\mathbf{M})=SPAN(B)$.
Finally we denote as $\mathbf{a}_{r}$ the instance of vector $\mathbf{a}$ at round $r$.  \\

We prove the bound for the family ${\cal M}($DBL$_{k})$ by first proving the lower  bound for ${\cal M}($DBL$_{2})$. Considering the latter proof, we first introduce a system of equations that characterizes the states of the nodes of the multigraph at round zero. The lower bound for ${\cal M}($DBL$_{2})$ is then proved by studying the evolution of this system of equations through the rounds.

\noindent{\bf Consider $M \in {\cal M}($DBL$_{2})$ and $r=0$:}  At the end of round $0$ the leader has state $S(v_l,0):[\{(1,[\bot]),(2,[\bot])\}]$, this leader state can be generated by many different configurations of nodes and edges in $M$. Such configurations are determined by the number of non leader processes with states
 $[\{1\}],[\{2\}],[\{1,2\}]$ and they are solutions of the following system of equations at round $0$:

\begin{equation} \label{system:ml1}
\begin{array}{l l}
\begin{cases}
    |(1,[\bot])|=|[\{1\}]|+|[\{1,2\}]| & \\
     |(2,[\bot])|=|[\{2\}]|+|[\{1,2\}]| \\
  \end{cases} \\
\scriptstyle{r=0}
 \end{array} 
\end{equation}
with the additional constraint that any variable in the solution cannot assume a negative value. When the leader updates its state, in successive rounds, we have a new  system of equations. The system of equations \ref{system:ml1} can be written in a matrix form as follows:  \begin{equation} \label{m1} \mathbf{m}_0=\mathbf{M}_{0}\mathbf{s}_0 \end{equation} where 
 $\mathbf{M}_{0}: \begin{bmatrix}
  1 & 0 & 1  \\
  0 & 1 & 1 \\
  \end{bmatrix}$ represents the matrix of coefficients of the system at round $0$,  $\mathbf{m}_0$ is the column vector of constant terms (each component of $\mathbf{m}_r$ represents the multiplicity of a certain element in the state of the leader at round $r$) and $\mathbf{s}_0$ is a solution vector.  Let us remark that  $\mathbf{M}_{r}$  depends of the round only while  $\mathbf{m}_r$ depends of  the leader state at round $r$. As a consequence  $\mathbf{M}_{r}$ characterizes any multigraph of the family ${\cal M}($DBL$_{2})$.
     
 The matrix $\mathbf{M}_{0}$ is characterized by $ker(\mathbf{M}_{0})= SPAN( \mathbf{k}_{0}:\begin{bmatrix}  1 & 1 & -1\end{bmatrix}^{\intercal} )$. Solutions of the matrix equation \ref{m1} are related by the following linear combination with the kernel vector $\mathbf{k}_0$:  $\mathbf{s}_0'=\mathbf{s}_0+ t \mathbf{k}_{0}$ with $t \in \mathbb{N}$ and such that each component of $\mathbf{s}'_0$ is non negative. As a consequence, given $\mathbf{m}_0$ the possible solutions of (\ref{system:ml1}) are restricted to a finite discrete set of points over a segment with direction $\mathbf{k}_{0}$. 
  From the point of view of the leader each solution represents a distinct graph belonging to $ {\cal M}(\text{DBL}_2)$ with a different number of processes: $\sum \mathbf{s}'_0 - \sum \mathbf{s}_0 =t \sum \mathbf{k}_0=t$. 
 

 Considering the example of Figure \ref{figure:indist},
the system of equations at round $0$ for the multigraph $M$ is the following
\begin{equation} \label{ex-system:ml1}
\begin{array}{l l}
\begin{cases}
    2=|[\{1\}]|+|[\{1,2\}]| & \\
    2=|[\{2\}]|+|[\{1,2\}]| \\
  \end{cases} \\
\scriptstyle{r=0}
 \end{array} 
\end{equation}

where
 $\mathbf{m}_0:\begin{bmatrix} 2 & 2 \end{bmatrix}^{\intercal}$. For such system of equations a solution is  $\mathbf{s}_0: \begin{bmatrix} 0 & 0 & 2 \end{bmatrix}^{\intercal}$, then using the kernel transformation another solution is $\mathbf{s}_0'=\mathbf{s}_0 +2 \mathbf{k_0}: \begin{bmatrix} 2 & 2 & 0\end{bmatrix}^{\intercal}$, these two solutions correspond to two $M,M' \in {\cal M}(\text{DBL}_2)$ of different size that generate the same state $S(v_l,0)$ as depicted in Figure \ref{figure:indist}. These two graphs are indistinguishable from the leader at round $0$ thus the leader is not able to output a correct count.


  
 \begin{figure}[htbp] 
\begin{center}
\includegraphics[scale=0.9]{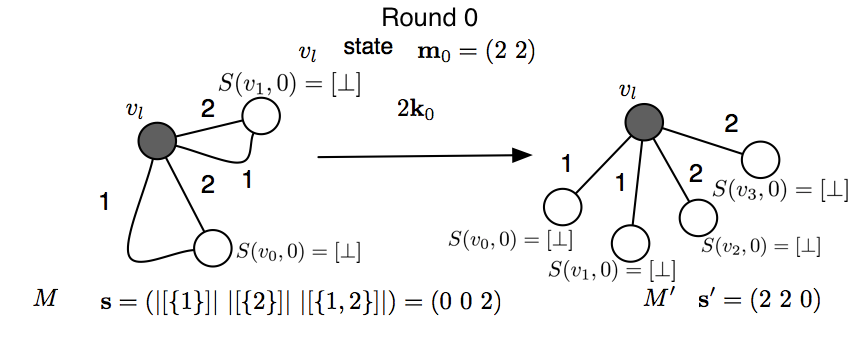}
\caption{Two dynamic multigraphs $ M, M' \in {\cal M}(\text{DBL}_{2})$ of different size that are indistinguishable at round $r=0$, the relationship among $M,M'$ is given by the kernel vector $\mathbf{k}_0$ \label{figure:indist} }

\end{center}

\end{figure}

The idea that we will use to show the lower bound is to characterize how the kernel space of $\mathbf{M}_r$ evolves and under which condition of the kernel space we have an unique solution, that corresponds to an unique size $|W|$. 

\paragraph{General structure of $\mathbf{M}_r$.}

At round $r$, the system of equations becomes $\mathbf{m}_r=\mathbf{M}_{r}\mathbf{s}_r$. 
The number of columns of $\mathbf{M}_r$ is equal to the number of all possible states of non-leader nodes, at round $r+1$ , which is  $column(r)=3^{r+1}$.
Each row of $\mathbf{M}_r$ corresponds to a possible connection $(j,S(v,r'))$ of $v_l$ at some round $0 \leq r' \leq r$. 
Thus the number of rows at round $r$ is two times the number of  existent states in $[0,r]$: $row(r)=2\sum_{k=0}^{r}3^k$.


As an example, at the end of round $1$, the system contains 8 equations ($2\cdot 3^0+2\cdot 3^1$) and 9 variables (i.e. $3^2$ rows) and the associated matrix $\mathbf{M}_{1}$ are:


\begin{figure}[H]
\hspace{-2cm}
\scriptsize
\subfloat{
\parbox{12cm}{
\begin{equation} \label{system:ml2}
\begin{array}{l l}
\begin{cases}
    |(1,[\bot])|=\sum_{\forall j \in \{\{1\},\{2\},\{1,2\} \}} |[\{1\},j]|+\sum_{\forall j \in \{\{1\},\{2\},\{1,2\} \}}|[\{1,2\},j]| & \\
     |(2,[\bot])|=\sum_{\forall j \in \{\{1\},\{2\},\{1,2\} \}}  |[\{2\},j]|+\sum_{\forall j \in \{\{1\},\{2\},\{1,2\} \}}|[\{1,2\},j]| & \\
          |(1,[\{1\}])|=|[\{1\},\{1\}]|+|[\{1\},\{1,2\}]| & \\
           |(1,[\{2\}])|=|[\{2\},\{1\}]|+|[\{2\},\{1,2\}]| & \\
 |(1,[\{1,2\}])|=|[\{1,2\},\{1\}]|+|[\{1,2\},\{1,2\}]| & \\
   |(2,[\{1\}])|=|[\{1\},\{2\}]|+|[\{1\},\{1,2\}]| & \\
   |(2,[\{2\}])|=|[\{2\},\{2\}]|+|[\{2\},\{1,2\}]| & \\
   |(2,[\{1,2\}])|=|[\{1,2\},\{2\}]|+|[\{1,2\},\{1,2\}]| 
   \end{cases}\\
   \scriptstyle{r=1}
  \end{array} 
\end{equation}}

} \subfloat{
\parbox{7cm}{
\begin{equation}\label{m2}
\mathbf{M}_{1}= \begin{bmatrix} 
1 &1 &1 &0 &0 &0 &1 &1 &1  \\
0 &0 &0 &1 &1 &1 &1 &1 &1  \\
1 &0 &1 &0 &0 &0 &0 &0 &0  \\
0 &0 &0 &1 &0 &1 &0 &0 &0  \\
0 &0 &0 &0 &0 &0 &1 &0 &1  \\
0 &1 &1 &0 &0 &0 &0 &0 &0  \\
0 &0 &0 &0 &1 &1 &0 &0 &0  \\
0 &0 &0 &0 &0 &0 &0 &1 &1  \\

\end{bmatrix} 
\end{equation} }
}

\end{figure}


Let us now consider how it is built the equation at round $r$ derived from the generic leader connection $(j,[x_0, \ldots, x_{r'-1}])$ with $j \in \{1,2\}$ introduced at round $r'$ in the system of equations, i.e.,$|(j,[x_0, \ldots, x_{r'-1}])|=|[x_0, \ldots, x_{r'-1},\{j\}]|+ |[x_0, \ldots, x_{r'-1},\{1,2\}]$. 
This equation at round $r$ becomes:
  \[  \quad
 \scriptsize 
 |(j,[ x_0,\ldots,x_{r'} ])|=\sum_{\forall s \in (\{1\}|\{2\} | \{1,2\})^{r-r'} } |[x_0,\ldots,x_{r'-1},\{j\},s]|+\sum_{\forall s \in (\{1\}|\{2\} | \{1,2\})^{r-r'}}[x_0,\ldots,x_{r'-1},\{1,2\}, s]|
 \]
 where $ (\{1\}|\{2\} | \{1,2\})^{r-r'}$ is the set of all possible lists with elements in $\{\{1\},\{2\},\{1,2\}\}$ and size $r-r'$.  As an example see the equation associated with  $|(1,[\bot])|$ at round $0$ (see  Equation \ref{system:ml1}) and the equation associated with  $|(1,[\bot])|$ at round $1$ (see Equation \ref{system:ml2}).

  Let notice  $ker(\mathbf{M}_{1})=\{\mathbf{k}_1 =\begin{bmatrix}
1 &1 &-1 &1 &1 &-1 &-1 &-1 &1
\end{bmatrix}^{\intercal}\}$, thus we have $<\mathbf{k}_1>=1$ with $<\mathbf{k}_1>^+=5,\,<\mathbf{k}_1>^-=4$. Now let us consider a solution $\mathbf{s}_1$ with $<\mathbf{s}_1> \leq 3$. It is easy to see that $\mathbf{s'}_1=\mathbf{s}_1+t\mathbf{k}_1$ has at least one negative component for any $t \neq 0$: since  $<\mathbf{k}_1>^-=4$  is not possible to have a solution $\mathbf{s}_1$  that as at least one unitary component for each negative component of $\mathbf{k}_1$. Thus $\mathbf{s'}_1$ cannot be a solution that represents a dynamic multigraph. 

This means that if $n \leq 3$ is possible to obtain the count in $2$ rounds, since there is only one possible solution of the system of equations for any $\mathbf{m}_1$ generated by a multigraph with $n \leq 3$. For $n \geq 4$ we have at least two possible solutions of different size, i.e. we have $\mathbf{s}_1=\begin{bmatrix}
0 &0 &1 &0 &0 &1 &1 &1 &0
\end{bmatrix}^{\intercal}$ with $n=4$ processes and $\mathbf{s}':\begin{bmatrix}
1 &1 &0 &1 &1 &0 &0 &0 &1
\end{bmatrix}^{\intercal}=\mathbf{s}_1+\mathbf{k}_1$ with $n=5$. It is easy to check that $\mathbf{m_1}=\mathbf{M}_{1}\mathbf{s}_1=\mathbf{M}_{1}\mathbf{s}_1'$, thus we have two multigraphs of different sizes that generate the same state $\mathbf{m}_1$ at $v_l$, see Figure \ref{figure:indist2}. \\

 \begin{figure}[htbp] 
 \hspace{-1cm}
\includegraphics[scale=0.45]{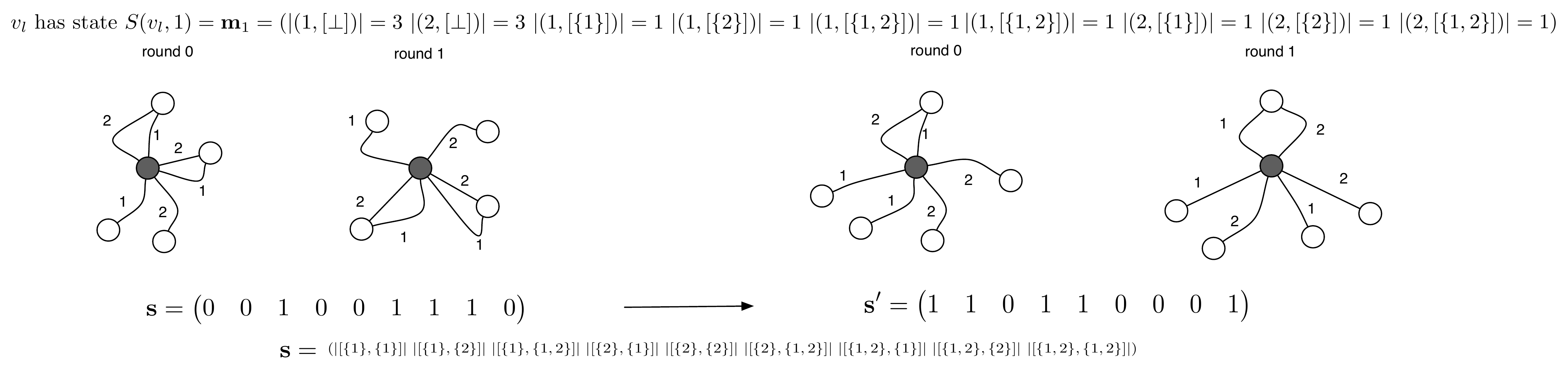}

\caption{Two dynamic multigraph $ M, M' \in {\cal M}(\text{DBL}_{2})$ of different size that are indistinguishable at round $r=1$, they induce the same leader state $S(v_l,1)=\mathbf{m}_1$, the relationship among the two is given by the kernel vector $\mathbf{k}_1$ \label{figure:indist2} }

\end{figure}

In order to simplify the proofs, we order columns of $\mathbf{M}_r$  lexicographically with respect to the state of a process. We consider the following order among elements $\{1\} < \{2\} <\{1,2\}$. As a consequence, the first column of $\mathbf{M}_{r}$ will correspond to state: $|[\{1\}, \ldots \{1\}]|$, the second column $|[\{1\}, \ldots \{1\},\{2\}]|$ and the last one $|[\{1,2\},\ldots,\{1,2\}]|$. Rows are ordered in the same way. This ordering has been used in Equation \ref{m1} and Equation \ref{m2}. Fixed this ordering we can use a connection $(j,[x_0,\ldots,x_{r'-1}])$ to denote a row $\mathbf{v}=(\mathbf{M}_{r})_{(j,[x_0,\ldots,x_{r'-1}])}$ and a node state to denote a single component of a vector, i.e. $(\mathbf{v})_{[x_0,\ldots,x_{r-1}]}$.
Moreover we have that the row vector $(\mathbf{M}_r)_{(j,[x_0,\ldots,x_{r'-1}])}$ will have two trails of ones, with length $3^{r-r'}$, for all columns in the form $|[x_0,\ldots,x_{r'-1},\{j\},s]|,|[x_0,\ldots,x_{r'-1},\{1,2\},s ]|$ with $ s \in (\{1\}|\{2\} | \{1,2\})^{r-r'}$, and zero for all the other columns (as reference see Eq. \ref{m2}).


\medskip

In the following lemmas we specifically characterize the structure of the kernel space of $\mathbf{M}_{r}$ in order to identify at which round there is an unique solution. 

\begin{lemma}\label{lemma:uniquekernel}
Let us consider the matrix $\mathbf{M}_r$ of the family  ${\cal M}($DBL$_{2})$ at round $r$. The dimension of the kernel space of  $\mathbf{M}_{r}$ is one (i.e., $ker(\mathbf{M}_{r})=SPAN(\mathbf{k}_{r})$).
\end{lemma}
\begin{proof} 
We first show that rows of $\mathbf{M}_{r}$ are linearly independent, thus that the rank of the matrix is equal to the number of rows. 
The proof is by induction:
\begin{itemize}
\item Base Case $r=0$:  $\mathbf{M}_{0}=\begin{bmatrix} 1 &0 & 1 \\ 0 & 1 & 1 \end{bmatrix}$, $det(\mathbf{M}_{0})=1$ thus the rows are linearly independent. 
\item Inductive Case $r$:  $\mathbf{M}_{r}$ can be written  as $\mathbf{M}_{r}= \begin{bmatrix} \mathbf{M}'_{r-1}  \\ \mathbf{U} \end{bmatrix}$, where $\mathbf{M}'_{r-1}$ is the matrix obtained by $\mathbf{M}_{r-1}$ substituting each element $1/(0)$ of $\mathbf{M}_{r-1}$ with a row vector $\begin{bmatrix} 1 & 1 & 1 \end{bmatrix} / (\begin{bmatrix} 0 & 0 & 0 \end{bmatrix} )$. Now by inductive hyp. we have that all rows of $\mathbf{M}_{r-1}$ are linearly independent. 
This implies that also the rows of $\mathbf{M}'_{r-1}$ are linearly independent, this can be easily shown by contradiction, let us suppose that we have $(\mathbf{M}'_{r-1})_{s}=x_{a}( \mathbf{M}'_{r-1})_{a}+x_b (\mathbf{M}'_{r-1})_{b}$ for some rows $s,a,b$ and two coefficient  $x_a,x_b$, this means that also if we take the subvectors $\mathbf{v}^{1},\mathbf{v}^{2},\mathbf{v}^{3}$ of $(\mathbf{M}'_{r-1})_{s},(\mathbf{M}'_{r-1})_{a},(\mathbf{M}'_{r-1})_{b}$, obtained by taking the components in position $j $ such that $jmod3=0$, we must have $\mathbf{v}^{1}=x_a\mathbf{v}^{2}+x_b \mathbf{v}^{3}$ but this could be also written as $(\mathbf{M}_{r-1})_{s}=x_{a} (\mathbf{M}_{r-1})_{a}+x_b( \mathbf{M}_{r-1})_{b}$ that is clearly a contradiction since the rows of $\mathbf{M}_{r-1}$ are linearly independent. We now show that a row of $\mathbf{U}$ cannot be expressed as linear combination of rows of $\mathbf{M}'_{r-1}$, we have that the row $\mathbf{U}_{c}$ corresponding to connection $c:(j, [x_0,\ldots,x_{r-1}])$,  has only two elements different from zero contained in the subvector $\begin{bmatrix} 1 &0&1\end{bmatrix}$ (if $j=1$) or $\begin{bmatrix} 0 &1&1\end{bmatrix}$ (if $j=2$) positioned in the columns with the form $[x_0,\ldots,x_{r-1},(\{1\}| \{2\} |\{1,2\})^{1} ]$. Now,  for each row $(\mathbf{M}'_{r-1})_i$ considering only the values of columns $[x_0,\ldots,x_{r-1},(\{1\}| \{2\} |\{1,2\})^{1} ]$, we get  a subvector that is either $\begin{bmatrix} 1 &1 &1 \end{bmatrix}$  or $\begin{bmatrix} 0 &0 &0 \end{bmatrix}$. Therefore it follows  that $\begin{bmatrix} 1 &0&1\end{bmatrix}$ or $\begin{bmatrix} 0 &1&1\end{bmatrix}$ cannot be expressed as linear combination of the row vectors of $\mathbf{M}'_{r-1}$. 

We have to show that the rows vector of $\mathbf{U}$ are linearly independent. If we consider the sets of 3 columns in the form $[x_0,\ldots,x_{r-1},(\{1\}| \{2\} |\{1,2\})^{1} ]$, only two rows have some elements different from zero: they are either $\begin{bmatrix} 1 &0&1\end{bmatrix}$ or $\begin{bmatrix} 0 &1&1\end{bmatrix}$ that are linearly independent. 
\end{itemize}

This implies, for the rank-nullity theorem \cite{linearalgebra}, that the size of the kernel is $|ker(\mathbf{M}_{r})|= column(r)-row(r)= 3^{r+1} - 2\sum_{k=0}^{r}3^k=1$.  
\end{proof}

\begin{lemma}\label{lemma:kernelvector}
Let us consider the matrix $\mathbf{M}_r$ of the family  ${\cal M}($DBL$_{2})$ at round $r$. We have \\
$\mathbf{k}_{r}=\begin{bmatrix}  \mathbf{k}_{r-1} &  \mathbf{k}_{r-1} & - \mathbf{k}_{r-1} \end{bmatrix}^{\intercal}$ with $\mathbf{k}_{-1}=1$. 
\end{lemma}
\begin{proof} The proof is done by induction: 
\begin{itemize}[noitemsep,nolistsep] 

\item Base Case, $round=0$. $\mathbf{k}_{0}=\begin{bmatrix} 1 & 1 &-1 \end{bmatrix}^{\intercal}$ implies $\mathbf{0}=\mathbf{M}_0\mathbf{k}_0$.\smallskip

\item Inductive Case, $round=r$. We assume $\mathbf{k}_{r-1}  = \begin{bmatrix}  \mathbf{k}_{r-2}  & \mathbf{k}_{r-2} &  - \mathbf{k}_{r-2} \end{bmatrix}^{\intercal}$.
We show a vector $\mathbf{k}$ such that its product for the rows of $\mathbf{M}_r$ corresponding to $c'=(j,l':[x_0 ,\ldots ,x_{r'-1}])$, with $r' < r$ and $j \in \{1,2\}$, is $0$. Then we show that the same holds for the remaining rows of $\mathbf{M}_r$ that corresponds to connection $c=(j,l:[x_0 ,\ldots ,x_{r-1}])$, and finally we show that $\mathbf{k}=\mathbf{k}_r=\begin{bmatrix}  \mathbf{k}_{r-1} &  \mathbf{k}_{r-1} & - \mathbf{k}_{r-1} \end{bmatrix}^{\intercal}$.

 Let us consider the row-vector product $(\mathbf{M}_{r-1})_{c'}\mathbf{k}_{r-1}$ at round $r-1$, by definition of kernel we have: 
 
 \begin{equation}\label{kernull}
 0=\sum_{\forall s \in (\{1\} |\{2\} |\{1,2\})^{r-r'-1} } (\mathbf{k}_{r-1})_{|[l',\{j\},s]|}+\sum_{\forall s \in (\{1\}| \{2\}|\{1,2\})^{r-r'-1} } (\mathbf{k}_{r-1})_{|[l',\{1,2\},s]|}
 \end{equation}

 Let us build a vector  $\mathbf{k}=\begin{bmatrix} (\mathbf{k}_{r-1})_{1}\mathbf{k}_{0}  & (\mathbf{k}_{r-1})_{2}\mathbf{k}_{0} & \ldots & (\mathbf{k}_{r-1})_{3^{r}}\mathbf{k}_{0} \end{bmatrix}^{\intercal} $ and let us examine the row-vector product $(\mathbf{M}_{r})_{c}\mathbf{k}$ :

 \begin{equation}\label{smrt}
  \hspace{-2cm}
  \quad  \sum_{\forall   s \in (\{1\} |\{2\} |\{1,2\})^{r-r'-1} }( (\mathbf{k})_{|[l',(j), s,\{1\}]|}+ (\mathbf{k})_{|[l',(j), s,\{2\}]|}+ (\mathbf{k})_{|[l',(j), s,\{1,2\}]|} )+\sum_{\forall  s \in (\{1\} |\{2\} |\{1,2\})^{r-r'}} (\mathbf{k})_{|[l',\{1,2\},s]|}
 \end{equation}
 the first term of Eq. \ref{smrt} can be expressed as follow
\begin{gather*}
  \hspace{-2cm}
 \sum_{\forall  s \in (\{1\} |\{2\} |\{1,2\})^{r-r'-1} }( (\mathbf{k})_{|[l',(j), s,\{1\}]|}+ (\mathbf{k})_{|[l',(j), s,\{2\}]|}+ (\mathbf{k})_{|[l',(j), s,\{1,2\}]|} )=\\ \sum_{\forall  s \in (\{1\} |\{2\} |\{1,2\})^{r-r'-1} }( (\mathbf{k}_{r-1})_{|[l',(j),s]|})((\mathbf{k}_0)_1+(\mathbf{k}_0)_2+(\mathbf{k}_0)_3 )=\\ \sum_{\forall  s \in (\{1\} |\{2\} |\{1,2\})^{r-r'-1} }(\mathbf{k}_{r-1})_{|[l',(j),s]|}
\end{gather*}
 
The second term  of Eq. \ref{smrt} can be rewritten as the first term, then by applying Eq. \ref{kernull}, we have: 
 
 $$ (\mathbf{M}_{r})_{c'}\mathbf{k} = (\mathbf{M}_{r-1})_{c'}\mathbf{k}_{r-1}=0$$

Now let us consider the row $c=(j,l:[x_0,\ldots,x_{r-1}])$ the row-vector product is $(\mathbf{M}_{r})_{c}\mathbf{k}$:
$$(\mathbf{k})_{|[l,\{j\}]|}+(\mathbf{k})_{|[l,\{1,2\}]|}=(\mathbf{k}_{r-1})_{|[l]|} ((\mathbf{k}_{0})_{j} +(\mathbf{k}_{0})_{3})=0$$


Thus we have $\mathbf{0}=\mathbf{M}_{r}\mathbf{k}$. Now we have to prove that $\mathbf{k}=\begin{bmatrix}  \mathbf{k}_{r-1} &  \mathbf{k}_{r-1} &- \mathbf{k}_{r-1} \end{bmatrix}^{\intercal}$.\\
 For inductive hypothesis we have $\mathbf{k}_{r-1}=\begin{bmatrix} \mathbf{k}_{r-2} & \mathbf{k}_{r-2} & - \mathbf{k}_{r-2} \end{bmatrix}^{\intercal}$, moreover we have, for Lemma \ref{lemma:uniquekernel}, that $\mathbf{k}_{r-1}=\begin{bmatrix} (\mathbf{k}_{r-2})_{1}\mathbf{k}_{0}  & (\mathbf{k}_{r-2})_{2}\mathbf{k}_{0} & \ldots & (\mathbf{k}_{r-2})_{3^{r-1}}\mathbf{k}_{0} \end{bmatrix}^{\intercal}$. 
This means  the first $3^{r}$ components of $\mathbf{k}'$ are $\mathbf{k}_{r-1}$, the same holds for the second $3^{r}$ components, and the last $3^{r}$ are $-\mathbf{k}_{r-1}$. This completes the proof. 

 \end{itemize}
\end{proof}

The proof of the following lemma can be found in the Appendix
\begin{lemma}\label{lemma:kernels2} Let us consider the matrix $\mathbf{M}_r$ of the family  ${\cal M}($DBL$_{2})$ at round $r$.\\ We have: $\begin{cases} min(\sum^{+} \mathbf{k}_{r}, \sum^{-} \mathbf{k}_{r} )=\sum^{-} \mathbf{k}_{r}=\frac{1}{2}(3^{r}+1)-1 & \\ \sum \mathbf{k}_{r} =1 \end{cases} $
\end{lemma}
\begin{proof}
Thanks to lemma \ref{lemma:kernelvector}, we have  $\sum^{+} \mathbf{k}_{r} =2 \sum^{+}\mathbf{k}_{r-1} + \sum^{-} \mathbf{k}_{r-1}$, $\sum^{-} \mathbf{k}_{r} =2 \sum^{-} \mathbf{k}_{r-1}  + \sum^{+}\mathbf{k}_{r-1}$ and $\sum^{+} \mathbf{k}_{0}=2$, $\sum^{-} \mathbf{k}_{0} =1$. We first prove, by induction, that $\sum \mathbf{k}_r=(\sum^{+} \mathbf{k}_{r} -\sum^{-} \mathbf{k}_{r})=1$.
\begin{itemize}[noitemsep,nolistsep]
 \item Base case \emph{ round=0}: $\sum^{+} \mathbf{k}_{0}- \sum^{-} \mathbf{k}_{0}=1$.
 \item Inductive Case \emph{ round=r:}  $(\sum^{+} \mathbf{k}_{r} - \sum^{-} \mathbf{k}_{r})= (2\sum^{+}\mathbf{k}_{r-1} + \sum^{-} \mathbf{k}_{r-1} - \sum^{+}\mathbf{k}_{r-1}  -2 \sum^{-} \mathbf{k}_{r-1} )=(\sum^{+}\mathbf{k}_{r-1} - \sum^{-} \mathbf{k}_{r-1})=1$.
 \end{itemize} 
 This result leads to the following recursive relation $\sum^{+} \mathbf{k}_r =3\sum^{+} \mathbf{k}_{r-1} -1$ for $\sum^{+} \mathbf{k}_r $ with base condition $\sum^{+} \mathbf{k}_0=2$. Solving the recursive relation we obtain $\sum^{+} \mathbf{k}_r=\frac{1}{2}(3^{r+1}+1)$. This implies  $min(\sum^{+} \mathbf{k}_{r} , \sum^{-} \mathbf{k}_{r} )=\frac{1}{2}(3^{r+1}+1) - 1$ where the last term takes into account the fact that the minimum is always the negative component and that $\sum \mathbf{k}_{r} =1$. 

\end{proof} 
From the previous lemma we have:

\begin{lemma}\label{lemma:bound}
Let us consider $M,M' \in {\cal M}($DBL$_{2})$ such that their sizes are: $|W|=n$ and $|W'|=n+1$. Does not exist an algorithm $\mathcal{A}_l$ that at round $r \leq \lfloor \text{log}_{3}(2|n|+1) \rfloor$ is able to distinguish if it is running on multigraph $M$ or $M'$.
\end{lemma}
\begin{proof}
Let us suppose by contradiction that such algorithm $\mathcal{A}_l$ exists. For lemma \ref{lemma:kernels2} we have  $ \sum^{-} \mathbf{k}_{r} = \frac{1}{2}(3^{r+1}+1)-1 \leq n$. Let us consider a configuration of non leader processes represented by vector $\mathbf{s}_{r}$ with $\sum \mathbf{s}_{r}= n$ and such that $(\mathbf{s}_{r})_j \geq 1$ for each $j \,| \, (\mathbf{k}_{r})_j < 0$ and $(\mathbf{s}_{r})_j =0$ otherwise. 
This implies there exists a dynamic multigraph $M:\{M_1,\ldots,M_r,\ldots \}$, of size $n$, obtained from $\mathbf{s}_r$ such that the leader state $S(v_l,r)$ at round $r$ is represented by $\mathbf{m}_r =\mathbf{M}_{r}\mathbf{s}_{r}$. Thus $\mathcal{A}_l$ outputs $n$  on the state $S(v_l,r)$.

Now let us consider $\mathbf{s'}_{r}=\mathbf{k}_{r}+ \mathbf{s}_{r}$, by construction we have $\forall j | (\mathbf{s'}_{r})_j>0$ thus $\mathbf{s'}_{r}$ represents an instance of dynamic multigraph $M':\{M'_1,\ldots,M'_r,\ldots \}$, let us denote $S'(v_l,r)$ the state of $v_l$ in $M'$. Since we have $\sum \mathbf{s'}_{r}=\sum \mathbf{k}_{r}+ \sum \mathbf{s}_{r}=n+1$, let us recall that  from Lemma \ref{lemma:kernels2} we have $\sum \mathbf{k}_{r}=1$, by hypothesis $\mathcal{A}_l$ outputs $n+1$  on the state $S'(v_l,r)$.

But by definition of kernel $\mathbf{M}_{r'}\mathbf{s}_{r}=\mathbf{M}_{r}\mathbf{s'}_{r}=\mathbf{m}_{r}$, thus $S(v_l,r)=S'(v_l,r)$ therefore $\mathcal{A}_l$  has to give the same output on the two different instances, that is a contradiction.

\end{proof}

From the previous lemma and considering that ${\cal M}($DBL$_2) \subseteq  {\cal M}($DBL$_k)$, we can state the following theorem whose proof is straightforward.

\begin{theorem}\label{lemma:dblk}
Any algorithm $\cal{A}$ cannot solve the counting on an instance $M \in {\cal M}($DBL$_k)$ at round $r < \lfloor \text{log}_{3}(2|W|+1) \rfloor-1$.
\end{theorem}

From Theorem \ref{lemma:dblk} and Lemma \ref{lemma:countingreduction} the next theorem immediately follows. 
\begin{theorem}\label{lemma:glp}
Given an instance $G \in {\cal G}$(PD$)_2$ any counting algorithm ${\cal A}$ on $G$ requires  $\Omega( \text{log}\,|V| )$ rounds.  
\end{theorem}

From the previous Theorem we have the following corollary.

\begin{corollary} \label{coro}
Given a dynamic network with dynamic diameter $D$, where $D$ is constant w.r.t. $|V|$. We have that any counting algorithm ${\cal A}$ requires at least $D+\Omega(\text{log}\,|V|)$ rounds. \label{diameter}
\end{corollary}
\begin{proof}
We create a configuration where $v_l$ is connected to two nodes $v_1,v_2$  by a static chain of $D-1$ nodes. Nodes $v_1,v_2$ are connected to the remaining ${\cal O}(|V|)$ nodes mimicking a ${\cal G}$(PD$)_2$ network. From this observation and Theorem  \ref{lemma:glp} the next corollary follows.
\end{proof}

\noindent \textbf{Discussion:} It is interesting to remark how the cost of counting is highly dependent on the information a node knows about the dynamic network (or about the adversary). Let us consider a simplified model where is not possible to have edges among nodes in the same set $V_h$, clearly this additional restriction does not impact on the validity of the lower bound and counting in ${\cal G}($PD$)_2$  still requires ${\Omega}(|V|)$ rounds. However, if we assume that a node $v$ knows its degree  $|N(v,r)|$ before the receive phase of round $r$,  as the model considered in \cite{ICDCS}, then counting in these restricted ${\cal G}($PD$)_2$ networks needs $O(1)$ rounds.  The algorithm is trivial: each node in $V_2$ sends a message $\frac{1}{|N(v,r)|}$ to nodes in $V_1$. A node in $V_1$ collects these messages and send their sum to the leader. 
This observation shows how a minimal knowledge on the dynamic graph, e.g., local degree known by a node at the beginning of round $r$, may have a great impact on the counting time.

\section{Conclusion}

This paper has shown that on anonymous dynamic networks with constant dynamic diameter there is at least a gap of $\Omega (\log\,|V|)$ rounds between counting and information dissemination even considering a model where the bandwidth is unlimited. 
 This difference is actually not present in static networks with IDs, dynamic network with IDs and anonymous static networks. In all these networks, when bandwidth is unlimited, there are algorithms showing that the cost of counting is at most the same order of the dynamic diameter \cite{1806760,sss13,MichailCS12}.
 On the contrary in anonymous networks with broadcast even when $D$ is constant  with respect to $V$, we proved that counting is still function of the network size. Solving counting in anonymous dynamic networks with broadcast is still an open problem (see the impossibility result conjectured in \cite{MichailCS12,sss13}) as well as finding a better bound than  $\Omega(D)$ for counting in anonymous 1-interval connected networks without any constraint on $D$. 

\medskip
\small
\noindent {\bf Acknowledgements:} The work has been partially supported by the Italian PRIN Project TENACE  "Protecting National Critical Infrastructures from Cyber Threats".

\bibliographystyle{plain}
\bibliography{biblio}

\end{document}